\documentclass[11pt,a4]{article}

\usepackage[USenglish]{babel}
\usepackage[utf8]{inputenc}
\usepackage[T1]{fontenc}
\usepackage[inline]{enumitem}
\usepackage{xspace}
\usepackage{amsmath}
\usepackage{amsthm}
\usepackage{amssymb}
\usepackage{xcolor}
\usepackage[colorlinks=True,
            citecolor=blue,
            urlcolor=violet,
            linkcolor=purple]{hyperref}
\usepackage{doi}
\usepackage[nameinlink, capitalize]{cleveref}
\usepackage[style=alphabetic, maxnames=99]{biblatex}
\usepackage{tikz}
\usepackage{pgfplots}
\usepackage{tcolorbox}
\usepackage{manyfoot}
\usepackage{algorithm}
\usepackage{algorithmicx}
\usepackage{algpseudocode}
\usepackage[top=25truemm,
			bottom=20truemm,
			left=20truemm,
			right=20truemm]{geometry}
\usepackage{titling}

\usetikzlibrary{plotmarks}
\usepgfplotslibrary{fillbetween}
\pgfplotsset{compat=1.16}

\tcbuselibrary{breakable}
\tcbuselibrary{skins}

\newcommand*{\quotient}[2]{\ensuremath{#1/\!\raisebox{-.40ex}{\ensuremath{\mathcal{#2}}}}}

\newtheorem{theorem}{Theorem}
\theoremstyle{plain}
\newtheorem{lemma}[theorem]{Lemma}
\newtheorem{conjecture}{Conjecture}

\newtheorem{proposition}[theorem]{Proposition}
\theoremstyle{definition}
\newtheorem{definition}[theorem]{Definition}

\def\True{{\textsc{True}}\xspace}
\def\False{{\textsc{False}}\xspace}
\def\UnfilledSet#1#2#3{{\mathrm{unfill}_{#3}(#1 \mid #2)}} 

\def\cc#1{\mathtt{#1}}
\DeclareMathOperator{\FNP}{\cc{FNP}}

\DeclareMathOperator{\TFNP}{\cc{TFNP}}
\DeclareMathOperator{\PPP}{\cc{PPP}}
\DeclareMathOperator{\PPAD}{\cc{PPAD}}

\DeclareMathOperator{\PPA}{\cc{PPA}}

\DeclareMathOperator{\PLS}{\cc{PLS}}
\DeclareMathOperator{\PWPP}{\cc{PWPP}}
\DeclareMathOperator{\EOPL}{\cc{EOPL}}

\DeclareMathOperator{\PLC}{\cc{PLC}}

\def\problem#1{\textsc{#1}\xspace}

\def\LongChoice{\problem{Long Choice}}
\def\EndofLine{\problem{End of Line}}
\def\EndofPotentialLine{\problem{End of Potential Line}}
\def\ConstrainedLongChoice{\problem{Constrained Long Choice}}

\def\UnaryLongChoice{\problem{Unary Long Choice}}
\def\QuotientPigeonCircuit{\problem{Quotient Pigeon}}
\def\PigeonCircuit{\problem{Pigeon}}
\def\LocalOPT{\problem{LocalOPT}}

\def\StonePicking{\problem{Interactive Bipartition Stone-Picking Game}}

\def\CheckSolutions{{\mathrm{CheckSolutions}}}
\title{
	Corrigendum: $\PLS$ is contained in $\PLC$
}

\author{
	Takashi Ishizuka\\
	Artificial Intelligence Laboratory, Fujitsu Limited, Japan\\
	\href{mailto:ishizuka-t@fujitsu.com}{\texttt{ishizuka-t@fujitsu.com}}
}

\begin{document}

\maketitle

%\begin{abstract}
%	Recently, \citeauthor{PPY23} \cite{PPY23} have introduced the new $\TFNP$ subclass called $\PLC$ that contains the class $\PPP$; they also have proven that several search problems related to extremal combinatorial principles (e.g., Ramsey's theorem and the Sunflower lemma) belong to $\PLC$.
%	This short paper shows that the class $\PLC$ also contains $\PLS$, a complexity class for $\TFNP$ problems that can be solved by a local search method.
%	However, it is still open whether $\PLC$ contains the class $\PPA$.
%\end{abstract}

\section*{Significant Error in the Proof of \cref{lemma:Q-Pigeon_to_LongChoice}}
We must apologize for the error in the proof of \cref{lemma:Q-Pigeon_to_LongChoice}.
Our reduction from \QuotientPigeonCircuit to \ConstrainedLongChoice does not work well. Therefore, we cannot guarantee the correctness of the main theorem: $\PLS$ is contained in $\PLC$.

\begin{conjecture}
	$\PLS$ is contained in $\PLC$.
\end{conjecture}

\paragraph{Details for the Error}
In our approach, we replace the original sequence with another sequence using the sub-procedure $\beta$.
After the replacement, we apply the protocol for finding a collision provided by \citeauthor{PPY23} \cite{PPY23}.

We tried to rely on the same property that \citeauthor{PPY23} \cite{PPY23} used in showing the $\PPP$-hardness of \LongChoice (see \cref{prop:key_properties_PPY23}).
However, we cannot guarantee that the first property (see \cref{PPY23_property1}) holds since we arranged the input sequence by the above replacement.
As a consequence, our proof does not work well.

\section*{The totality of \QuotientPigeonCircuit}
We show that the search problem \QuotientPigeonCircuit is a $\TFNP$ problem. total search problem. In particular, we prove that this problem belongs to $\TFNP$.

\begin{proof}
	Since the correctness of an obtained solution is verifiable in polynomial time, it suffices to show that every \QuotientPigeonCircuit instance has at least one solution.
	Let two Boolean circuits $C: [2^n] \to [2^n]$ and $E: [2^n] \times [2^n] \to \{ 0, 1 \}$ and a special element $v^{*} \in [2^n]$ be an instance of \QuotientPigeonCircuit.
	We can check whether the Boolean circuit $E$ computes an equivalence relation over the finite set $[2^n]$ in finite time. Hence, there is no solution to the violation of the equivalence relation without loss of generality.
	
	Consider the following sequence: $u_{0} := v^{*}$ and $u_{i} := C(u_{i-1})$ for every $i \ge 1$.
	We show that there is at least one of a pair of finite integers $i$ and $j$ such that $u_{i} \not\sim_{E} u_{j}$ and $C(u_{i}) \sim_{E} C(u_{j})$ and a finite integer $i$ such that $C(u_{i}) \sim_{E} v^{*}$.
	From our assumption, $x \not\sim_{E} y$ implies that $x \neq y$. Hence, there is a finite $i$ such that $u_{i} \sim_{E} u_{j}$ for some non-negative $j < i$. Let denote $i$ the smallest such an index. Thus, we have another index $j < i$ such that $u_{j} \sim_{E} u_{i}$, and it holds that $u_{k} \not\sim_{E} u_{\ell}$ for $0 \le k < \ell < i$.
	If $j = 0$, then $u_{i-1}$ is a solution to \QuotientPigeonCircuit since $C(u_{i-1}) = u_{i} \sim_{E} u_{j} = v^{*}$. Otherwise, a pair of distinct elements $u_{i-1}$ and $u_{j-1}$ is a solution to \QuotientPigeonCircuit since it staisfies that $u_{i-1} \not\sim_{E} u_{j-1}$ and $C(u_{i-1}) = u_{i} \sim_{E} u_{j} = C(u_{j-1})$.
	
	Therefore, the problem \QuotientPigeonCircuit is a $\TFNP$ problem.
\end{proof}

\section*{Acknowledgement}
We are filled with grateful thanks towards Noah Fleming who kindly pointed out the significant error in our proof.

\newpage

\title{
%	On the Complexity of \QuotientPigeonCircuit
	$\PLS$ is contained in $\PLC$
}

\author{
	Takashi Ishizuka\\
	Artificial Intelligence Laboratory, Fujitsu Limited, Japan\\
	\href{mailto:ishizuka-t@fujitsu.com}{\texttt{ishizuka-t@fujitsu.com}}
}

\maketitle

\begin{abstract}
	Recently, \citeauthor{PPY23} \cite{PPY23} have introduced the new $\TFNP$ subclass called $\PLC$ that contains the class $\PPP$; they also have proven that several search problems related to extremal combinatorial principles (e.g., Ramsey's theorem and the Sunflower lemma) belong to $\PLC$.
	This short paper shows that the class $\PLC$ also contains $\PLS$, a complexity class for $\TFNP$ problems that can be solved by a local search method.
	However, it is still open whether $\PLC$ contains the class $\PPA$.
\end{abstract}

\section{Introduction}
\subsection{Notation}
First of all, we present terminologies that we will use in this short paper.

We denote by $\mathbb{Z}$ the set of all integers.
For an integer $a \in \mathbb{Z}$, we define $\mathbb{Z}_{\ge a} := \{ x \in \mathbb{Z} ~:~ x \ge a \}$ and $\mathbb{Z}_{> a} := \{ x \in \mathbb{Z} ~:~ x > a\}$.
We use $[n] := \{ 1, 2, \dots, n \}$ for every positive integer $n$ in $\mathbb{Z}_{> 0}$.
Let $X$ be a finite set. We denote by $|X|$ the cardinality of the elements in $X$.
For any function $f: X \to X$ and any sequence of elements $\xi_{0}, \dots, \xi_{k}$ in $X$, the {\it unfilled set} of $X$ is defined as $X \setminus \{ f(\xi_{0}), \dots, f(\xi_{k}) \}$; we write this for $\UnfilledSet{X}{\xi_{0}, \dots, \xi_{k}}{f}$. When $X$ is a finite set of integers, for a positive integer $\kappa$, we denote $X[k]$ to be the set of $\kappa$ smallest elements of $X$; that is, $|X[\kappa]| = \kappa$ and $\xi < \eta$ for each pair of two elements $\xi \in X[\kappa]$ and $\eta \in X \setminus X[\kappa]$.

Let $\{ 0, 1 \}^{*}$ denote the set of binary strings with a finite length. 
For every string $x \in \{ 0, 1 \}^{*}$, we denote by $|x|$ the length of $x$.
For each positive integer $n$, we write $\{ 0, 1 \}^{n}$ for the set of binary strength with the length $n$.
Throughout this short paper, we sometimes regard $\{ 0, 1 \}^n$ as the set of positve integers $[2^n]$.

\paragraph{Search Problems}
Let $R \subseteq \{ 0, 1 \}^* \times \{ 0, 1 \}^*$ be a relation.
We say that $R$ is {\it polynomially balanced} if there is a polynomial $p: \mathbb{Z}_{\ge 0} \to \mathbb{Z}_{\ge 0}$ such that for each $(x, y) \in R$, it holds that $|y| \le p(|x|)$.
We say that $R$ is {\it polynomial-time decidable} if for each pair of strings $(x, y) \in \{ 0, 1 \}^* \times \{ 0, 1 \}^*$, we can decide whether $(x, y)$ belongs to $R$ in polynomial time.
We say that $R$ is {\it total} if for every string $x \in \{ 0, 1 \}^*$, there always exists at least one string $y$ such that $(x, y) \in R$.

For a relation $R \subseteq \{ 0, 1 \}^* \times \{ 0, 1 \}^*$, the search problem with respect to $R$ is defined as follows\footnote{For simplicity, we call the search problem with respect to $R$ the search problem $R$.}: Given a string $x \in \{ 0, 1 \}^*$, find a string $y \in \{ 0, 1 \}^*$ such that $(x, y) \in R$ if such a $y$ exists, otherwise reports ``{\it no}.''
When $R$ is also total, we call such a search problem a total search problem.
The complexity class $\FNP$ is the set of all search problems with respect to a polynomially balanced and polynomial-time decidable relation $R$.
The complexity class $\TFNP$ is the set of all total search problems belonging to $\FNP$. By definition, it holds that $\TFNP \subseteq \FNP$.

\paragraph{Reductions}
Let $R, S \subseteq \{ 0, 1 \}^* \times \{ 0, 1 \}^*$ be two search problems.
A polynomial-time reduction from $R$ to $S$ is defined by two polynomial-time computable functions $f: \{ 0, 1 \}^* \to \{ 0, 1 \}^*$ and $g: \{ 0, 1 \}^* \times \{ 0, 1 \}^* \to \{ 0, 1 \}^*$ satisfying that $(x, g(x, y)) \in R$ whenever $(f(x), y) \in S$.
In other words, the function $f$ maps an instance $x$ of $R$ to an instance $f(x)$ of $S$, and the other function $g$ maps a solution $y$ to the instance $f(x)$ to a solution $g(x, y)$ to the instance $x$.

For a complexity class $\mathcal{C}$, we say that a search problem $R$ is $\mathcal{C}$-hard if all search problems in $\mathcal{C}$ are polynomial-time reducible to $R$.
Furthermore, we say that a search problem $R$ is $\mathcal{C}$-complete if $R$ is $\mathcal{C}$-hard, and $R$ belongs to $\mathcal{C}$.

\subsection{Backgrounds}
Consider the following two-player game:
There are $2^n$ stones; we denote by $U_{0}$ the set of all stones.
In the first round, Player $1$ chooses one stone $a_{0}$ from $U_{0}$, then Player $2$ partitions remaining stones $U_{1} := U_{0} \setminus \{ a_{0} \}$ into two groups, denoted by $U_{1}^{0}$ and $U_{1}^{1}$.
In the second round, Player $1$ chooses one stone $a_{1}$ from $U_{1}^{b_1}$, then all stones in the opposite to $a_{1}$ (i.e., all stones in $U_{1}^{1-b_1}$) are removed from the game immediately. Player $2$ partitions all stones in the group $U_{2} := U_{1}^{b_1} \setminus \{ a_{1} \}$ into two groups, denoted by $U_{2}^{0}$ and $U_{2}^{1}$.
In the $i$-th round, Player $1$ chooses one stone $a_{i-1}$ from $U_{i-1}^{b_{i-1}}$, then all stones in $U_{i-1}^{1 - b_{i-1}}$ are removed from the game immediately. Player $2$ partitions all stones in the group $U_{i} := U_{i}^{b_{i-1}} \setminus \{ a_{i-1} \}$ into two groups $U_{i}^{0}$ and $U_{i}^{1}$.
They repeat such processes $n+1$ rounds.
Player $1$ wins if they can pick $n+1$ distinct stones $a_{0}, a_{1}, \dots, a_{n}$ at the end of the game. If Player $1$ cannot choose any stones during the above game, then Player $2$ wins.
In this paper, we call this game the \StonePicking\footnote{Remark that the term ``bipartition'' implies that, at each turn, Player $2$ partitions the current set into two groups. This game can be easily generalized to the multi-partition setting.}.

It is straightforward to see that a winning strategy for Player $1$ in \StonePicking always exists.
Recently, \citeauthor{PPY23} \cite{PPY23} have formulated the problem of finding a winning strategy for Player $1$ in \StonePicking as a $\TFNP$ problem.
Informally speaking, their problem, called \LongChoice, is to find a sequence of distinct elements that satisfies suitable properties when we are given a description of Player $2$'s action at each round.
The formal definition can be found in \cref{def:longchoice}.

$\TFNP$ problems \cite{MP91, Pap94} --- the existence of solutions guarantees, and the correctness of every solution is effortlessly checkable --- comprise a fascinating field in computational complexity theory.
It is known that many significantly important computational problems belong to the complexity class $\TFNP$. For example, finding a Nash equilibrium \cite{CDT09, DGP09}, computing a fair division \cite{FRG18, DFM22, GHH23}, integer factoring \cite{BO06, Jer16}, and algebraic problems related to cryptographies \cite{SZZ18, HV21}.
A natural way to analyze the theoretical features of a complexity class is to characterize its class by complete problems.
However, it is widely believed that $\TFNP$ has no complete problem \cite{Pud15, Pap94}.
Consequently, several $\TFNP$ subclasses with complete problems have been introduced over the past three decades.
The best well-known such classes include $\PLS$ \cite{JPY88}, $\PPAD$, $\PPA$, $\PPP$ \cite{Pap94}, $\PWPP$ \cite{Jer16}, and $\EOPL$ \cite{DP11, FGMS20, GH0MPRT22}. 
%see also \cite{GP17}.

We are interested in the boundary of total search problems.
In particular, our central motive is to capture the most hard problems among {\it syntactic} $\TFNP$ problems.
Previously, \citeauthor{GP18} \cite{GP18} have introduced a $\TFNP$ problem that unifies the traditional $\TFNP$ subclasses.
However, we are unaware of another $\TFNP$ problem that unifies $\PLS$, $\PPP$, and $\PPA$.
As a first step, this short paper sheds light on the relationship between the oldest and the newest $\TFNP$ subclasses.

\subsection{Our Contributions}
We make clear the relationship between two $\TFNP$ subclasses: $\PLS$ and $\PLC$.

The complexity class $\PLS$, introduced by \citeauthor{JPY88} \cite{JPY88}, is one of the most famous $\TFNP$ subclasses.
%This class was introduced by \citeauthor{JPY88} \cite{JPY88}.
The class $\PLS$ captures the complexity of the problems that can be solved by a local search method.
Formally, this class is defined as the set of all search problems that are reducible to \LocalOPT in polynomial time. \\

\begin{tcolorbox}
	\begin{definition}\label{def:localopt}
		\LocalOPT
		
		\noindent
		\textbf{Input}: 
		Two Boolean circuits $f : [2^n] \to [2^n]$ and $p: [2^n] \to [2^m]$.
		
		\noindent
		\textbf{Output}: 
		A point $x$ in $[2^n]$ such that $p(x) \ge p(f(x))$
	\end{definition}
\end{tcolorbox}

\begin{definition} \label{def:PLS}
	The complexity class $\PLS$ is the set of all search problems that are reducible to \LocalOPT in polynomial time.
\end{definition}

On the other hand, the complexity class $\PLC$ is the newest $\TFNP$ subclass, introduced by \citeauthor{PPY23} \cite{PPY23}.
This class is formulated as the set of all search problems that are reducible to the problem \LongChoice in polynomial time.\\

\begin{tcolorbox}
	\begin{definition}\label{def:longchoice}
		\LongChoice
		
		\noindent
		\textbf{Input}: 
		$n - 1$ Boolean circuits $P_{0}, \dots, P_{n-2}$ such that $P_{i}: ([2^n])^{i+2} \to \{0, 1\}$ for each $i \in \{0, \dots, n-2\}$ 
		
		\noindent
		\textbf{Output}: 
		A sequence of $n+1$ distinct elements $a_{0}, a_{1}, \dots, a_{n}$ in $[2^n]$ such that for each $i \in \{0, \dots, n-2\}$, $P_{i}(a_{0} \dots, a_{i}, a_{j})$ is the same for every $j > i$.
	\end{definition}
\end{tcolorbox}

\begin{definition} \label{def:PLC}
	The complexity class $\PLC$ is the set of all search problems that are reducible to \LongChoice in polynomial time.
\end{definition}

As mentioned before, the class $\PLC$ captures the complexity of finding a winning strategy for Player $1$ in \problem{Interactive Bipartition Stone-Picking Game}.
Let $P_{0}, \dots, P_{n-2}$ be a sequence of the predicates given by an instance of \LongChoice.
For each index $i \in \{ 0, \dots, n-2 \}$, the predicate $P_{i}: ([2^n])^{i+2} \to \{ 0, 1 \}$ represents Player $2$'s behavior at the $(i+1)$th round.
Then, we can easily regard a solution to \LongChoice as a winning strategy for Player $1$.

We show that the complexity class $\PLS$ is contained in $\PLC$.

\begin{theorem}[Main Contribution] \label{thm:main}
	$\PLS$ is contained in $\PLC$.	
\end{theorem}

\section{Technical Ingredients}
To prove our main theorem, we introduce a search problem, an extension of \PigeonCircuit (see \cref{def:pigeoncircuit}).
The complexity class $\PPP$, introduced by Papadimitriou \cite{Pap94}, is the class for search problems related to the pigeonhole principle. A total search problem belonging to $\PPP$ is to find a collision under the self-mapping.
For instance, consider the situation where we put $2^n$ pigeons in $2^n$ cages according to a function $C: [2^n] \to [2^n]$. Unfortunately, one of these cages is broken, denoted by $v^*$, and we cannot use it. The task of the problem is to find a collision or detect the broken cage being used.

The formal definition of the canonical $\PPP$-complete problem is as follows. \\

\begin{tcolorbox}
	\begin{definition}\label{def:pigeoncircuit}
		\PigeonCircuit
		
		\noindent
		\textbf{Input}: 
		A Boolean circuit $C : [2^n] \to [2^n]$
		and a special element $v^* \in [2^n]$
		
		\noindent
		\textbf{Output}: One of the following
		\begin{enumerate}
			\item two distinct elements $x, y \in [2^n]$ such that $C(x) = C(y)$
			\item an element $x \in [2^n]$ such that $C(x) = v^*$
		\end{enumerate}
	\end{definition}
\end{tcolorbox}

\begin{definition} \label{def:PPP}
	The complexity class $\PPP$ is the set of all search problems that are reducible to \PigeonCircuit in polynomial time.
\end{definition}

\begin{theorem}[\citeauthor{PPY23} \cite{PPY23}]
	$\PPP$ is contained in $\PLC$.
\end{theorem}

From now on, we extend \PigeonCircuit to another $\TFNP$ problem called \QuotientPigeonCircuit.
Let $U$ be a finite set, and let $\sim$ denote an equivalence relation over $U$.
We now consider a \PigeonCircuit instance over the quotient set $\quotient{U}{\sim}$.
In other words, we focus on the following search problem: Given a function $C: \quotient{U}{\sim} \to \quotient{U}{\sim}$ and a special element $v^*$ in $U$, find two distinct elements $x, y \in \quotient{U}{\sim}$ such that $C(x) \sim C(y)$ or an element $x \in \quotient{U}{\sim} $ such that $C(x) \sim v^*$.
For example, consider the situation where we put $N$ books away into $M$ bookshelves, but one of these bookshelves is broken; namely, this one cannot be used. Our behavior can be represented by a function $C: [N] \to [M]$.
We sort these books by genre, which can be classified into exactly $M$ genres. Note that genre induces an equivalence relation $E$ over these books. The task of the problem is to detect that two books of different genres are stored on the same bookshelf or that the broken bookshelf is being used.

To formulate the above variant of \PigeonCircuit, we allow to obtain another function $E: U \times U \to \{ 0, 1 \}$ computing an equivalence relation over $U$.
We denote by $\sim_{E}$ the binary relation defined by $E$; for each pair of elements $x, y$ in $U$, $x \sim_E y$ if and only if $E(x, y) = 1$.
Formally, the new search problem called \QuotientPigeonCircuit is defined as follows. \\

\begin{tcolorbox}
	\begin{definition}\label{def:quotientpigeoncircuit}
		\QuotientPigeonCircuit
		
		\noindent
		\textbf{Input}: 
		Two Boolean circuits $C : [2^n] \to [2^n]$ and $E: [2^n] \times [2^n] \to \{0, 1\}$
		and an element $v^* \in [2^n]$
		
		\noindent
		\textbf{Ouput}: 
		One of the following
		\begin{enumerate}
			\item two elements $x, y \in [2^n]$ such that $x \not\sim_{E} y$ and $C(x) \sim_{E} C(y)$
			\item an element $x \in [2^n]$ such that $C(x) \sim_{E} v^*$
			\item two elements $x, y \in [2^n]$ such that $x \sim_{E} y$ and $C(x) \not\sim_{E} C(y)$
			\item an element $x \in [2^n]$ such that $E(x, x) = 0$
			\item two elements  $x, y \in [2^n]$ such that $E(x, y) \neq E(y, x)$.
			\item three distinct elements $x, y, z \in [2^n]$ such that $x \sim_E y$, $y \sim_E z$, and $x \not\sim_E z$
		\end{enumerate}
	\end{definition}
\end{tcolorbox}

Unfortunately, we are unaware of a way of syntactically enforcing the Boolean circuit $E$ to compute an equivalence relation over the finite set $[2^n]$.
Thus, we introduce violations as solutions to \QuotientPigeonCircuit to ensure that this problem belongs to $\TFNP$.
More precisely, the fourth-type solution is a violation of the {\it reflexivity}.
The fifth-type solution represents a violation of the {\it symmetry}.
Finally, the sixth-type solution means a violation of the {\it transivity}.

\begin{proposition}
	\QuotientPigeonCircuit is $\PPP$-hard.
\end{proposition}
\begin{proof}
	It suffices to define the Boolean circuit $E: [2^n] \times [2^n] \to \{0, 1\}$ as $E(x, y) = 1$ if and only if $x = y$ for all $x, y \in [2^n]$.
\end{proof}

Before closing this section, we observe the useful properties of \QuotientPigeonCircuit.
First, we show that we can assume that there is no fixed point for any \QuotientPigeonCircuit without loss of generality.

\begin{proposition} \label{prop:no-fixedpoint}
	Let $\langle C: [2^n] \to [2^n], E: [2^n] \times [2^n] \to \{0, 1\}, v^* \in [2^n] \rangle$ be an instance of \QuotientPigeonCircuit.
	We can assume that there is no element $x \in [2^n]$ such that $C(x) \sim_{E} x$ or $C(x) = v^*$ without loss of generality.
\end{proposition}
\begin{proof}
	We first redefine the Boolean circuit $C$ as follows: For each element $x \in [2^n]$, $C(x) := u^*$ if $C(x) \sim_{E} v^*$, otherwise we do not modify the output of $C(x)$. Here, we pick arbitrary element $u^*$ in $[2^n]$ with $u^* \neq v^*$.
	It is straightforward to see that we can recover a solution to the original instance from a solution to the modified instance.
	
	We write $U$ for the set $\{ 0, 1 \} \times [2^n]$.
	We construct new two Boolean circuits $C': U \to U$ and $E': U \times U \to \{ 0, 1 \}$ as follows.
	For each element $(b, x) \in U$, we define $C'(b, x) := ( 1-b, C(x) )$.
	For each pair of two elements $(b, x)$ and $(c, y)$ in $U$, define $E'((b, x), (c, y)) = 1$ if and only if $b = c$ and $E(x, y) = 1$.
	Informally speaking, we create a copy of equivalent classes induced by $E$.
	Finally, we define $(0, v^*)$ in $U$ to be a special element for the reduced instance of \QuotientPigeonCircuit.
	
	Since the first bit is always flipped by the Boolean function $C'$, there is no fixed point of $C'$. Furthermore, it holds that $\xi \not\sim_{E'} \eta$ for all elements $\xi, \eta \in U$ whose first bits are different.
	This implies that there is no element $\xi \in U$ such that $C'(x) \sim_{E}$.
	
	From our construction, it is easy to see that we can efficiently recover an original solution from a solution to the reduced instance.
\end{proof}

Next, we show that we can suppose that there exist at least $2n$ equivalent classes without loss of generality.
More precisely, we can assume that the elements $u_{0} := v^{*}, u_{1} := C(u_{0}), \dots, u_{i} := C(u_{i-1}), \dots, u_{2n}$ are distinct under $\sim_{E}$, i.e., $u_{i} \not\sim_{E} u_{j}$ for all $0 \le i < j \le 2n$, without loss of generality.

\begin{proposition} \label{prop:assumption_for_quotient-pigeon}
	For every non-trivial\footnote{We say that an instance is {\it trivial} if it can be solved by a n\"{i}ve approach in polynomial time.} instance of \QuotientPigeonCircuit $\langle C: [2^n] \to [2^n], E: [2^n] \times [2^n] \to \{0, 1\}, v^* \in [2^n] \rangle$, we can assume that the elements $u_{0} := v^{*}, u_{1} := C(u_{0}), \dots, u_{i} := C(u_{i-1}), \dots, u_{2n}$ are distinct under $\sim_{E}$, i.e., $u_{i} \not\sim_{E} u_{j}$ for all $0 \le i < j \le 2n$, without loss of generality.
\end{proposition}
\begin{proof}
	It is sufficient to prove that we can easily recover a solution to the original \QuotientPigeonCircuit instance when there are two elements  $u_{i}$ and $u_{j}$ such that $u_{i} \sim_{E} u_{j}$.
	We prove this fact by induction.
	
	In the base case, we can assume that $u_{0} \not\sim_{E} u_{1} = C(u_{0})$ from \cref{prop:no-fixedpoint}.
	
	In the inductive case, suppose that the elements $u_{0} := v^{*}, u_{1} := C(u_{0}), \dots, u_{i} := C(u_{i-1})$ are distinct under $\sim_{E}$ for some positive integer $i < 2n$.
	If the element $u_{i+1} := C(u_{i})$ collided with an element $u_{j}$ under $\sim_{E}$ for some $0 \le j \le i$, we can effortlessly recover a solution to the original \QuotientPigeonCircuit instance as follows:
	The element $u_{i}$ is the second-type solution when $j = 0$; and the elements $u_{i}, u_{j-1}$ is the first-type solution since $u_{i} \not\sim_{E} u_{j-1}$ and $C(u_{i}) \sim_{E} u_{j-1}$ when $j > 0$.
\end{proof}

\section{Proof of Our Main Theorem} \label{Section:Proof_of_Main_Theorem}
\cref{thm:main} immediately follows from the following two lemmata.

\begin{lemma} \label{lemma:PLS-hardness_of_Q-Pigeon}
	\QuotientPigeonCircuit is $\PLS$-hard.
\end{lemma}

\begin{lemma} \label{lemma:Q-Pigeon_to_LongChoice}
	\QuotientPigeonCircuit belongs to $\PLC$.
\end{lemma}

In other words, we first prove that there is a polynomial-time reduction from \LocalOPT to \QuotientPigeonCircuit in \cref{lemma:PLS-hardness_of_Q-Pigeon}.
After that, we show a polynomial-time reduction \QuotientPigeonCircuit to \LongChoice in \cref{lemma:Q-Pigeon_to_LongChoice}.
By the transitivity of the polynomial-time reduction, we have a polynomial-time reduction from \LocalOPT to \LongChoice. This implies that \LocalOPT belongs to the complexity class $\PLC$; therefore, we conclude our main theorem: $\PLS \subseteq \PLC$.

The proofs of \cref{lemma:PLS-hardness_of_Q-Pigeon} and \cref{lemma:Q-Pigeon_to_LongChoice} can be found in \cref{proof:PLS-hardness_of_Q-Pigeon} and \cref{proof:Q-Pigeon_to_LongChoice}, respectively.

% Proof PLS-hardness of Quotient Pigeon
\subsection{Proof of \cref{lemma:PLS-hardness_of_Q-Pigeon}} \label{proof:PLS-hardness_of_Q-Pigeon}
To prove this lemma, we show a polynomial-time reduction from \LocalOPT to \QuotientPigeonCircuit.
Our proof is inspired by the robustness proof of \EndofPotentialLine by \citeauthor{Ish21} \cite{Ish21}.

Let two Boolean circuits $f : [2^n] \to [2^n]$ and $p: [2^n] \to [2^m]$ be an instance of \LocalOPT.
We first show, in \cref{prop:assumption_for_localopt}, that we can assume that the point $1 \in [2^n]$ has the unit potential, and every point $x \in [2^n]$ with $f(x) \neq x$ satisfies that $p(f(x)) = p(x) + 1$, without loss of generality.

\begin{proposition} \label{prop:assumption_for_localopt}
	For every \LocalOPT instance $\langle f : [2^n] \to [2^n], p: [2^n] \to [2^m] \rangle$, we have a polynomial-time reduction from $\langle f, p \rangle$ to a \LocalOPT instance $\langle F: [2^m] \times [2^n] \to [2^m] \times [2^n], P: [2^m] \times [2^n] \to [2^m] \rangle$ that satisfies the following two properties:
	\begin{enumerate*}[label = (\arabic*)]
		\item For each point $\xi$ in $[2^m] \times [2^n]$ with $F(\xi) \neq \xi$, it holds that $P(F(\xi)) = P(\xi) + 1$; and
		\item we know a special point $v^*$ in $[2^m] \times [2^n]$ with $P(v^*) = 1$.
	\end{enumerate*}
\end{proposition}
\begin{proof}
	Let two Boolean circuits $f: [2^n] \to [2^n]$ and $p: [2^n] \to [2^m]$ be an instance of \LocalOPT.
	We first reduce the above instance to another \LocalOPT instance $\langle f': [2^n] \to [2^n], p': [2^n] \to [2^m] \rangle$ satisfying that $p'(1) = 1$.
	For each point $x \in [2^n]$, we define
	\begin{align*}
		f'(x) := \begin{cases}
			f(1) & \text{ if } f(x) = 1 \\
			f(x) & \text{ otherwise,}
		\end{cases}
	\end{align*}
	
	\noindent
	and $p'(x) = p(x)$ if $x \neq 1$, otherwise $p'(x) = 1$.
	It is easy to see that the instance $\langle f', p' \rangle$ holds the desired condition.
	
	What remains is to prove that we can recover a solution to the original instance $\langle f, p \rangle$ from a solution to the new instance $\langle f', p' \rangle$ in polynomial time.
	Let $x \in [2^n]$ be a solution to $\langle f', p' \rangle$; that is, it holds that $p'(x) \ge p'(f'(x))$.
	First, we suppose that $x = 1$. This implies that $p(x) \ge p'(x) = 1 = p'(f'(x)) = p(f(x))$. Hence, the point $x = 1$ is a solution to the original instance.
	Next, we suppose that $x > 1$, $f(x) = 1$, and the special point $1$ is not a solution to $\langle f', p' \rangle$. This implies that $p(x) = p'(x) \ge p'(f'(x)) = p(f(1))$. Therefore, we can see that at least one of $x$ and $f(x)$ is a solution to the original instance $\langle f, p \rangle$.
	Finally, we suppose that $x > 1$ and $f(x) \neq 1$. This implies that $p(x) = p'(x) \ge p'(f'(x)) = p(f(x))$. Thus, the point $x$ is a solution to the original instance.

	We move on to proving another desired condition: For every point, its potential increases at most one.
	We will construct another \LocalOPT instance $\langle F: [2^m] \times [2^n] \to [2^m] \times [2^n] , P: [2^m] \times [2^n] \to [2^m] \rangle$ from the \LocalOPT instance $\langle f': [2^n] \to [2^n], p': [2^n] \to [2^m] \rangle$ such that $P(F(i, x)) = P(i, x) + 1$ for each point $(i, x) \in [2^m] \times [2^n]$ with $F(i, x) \neq (i, x)$.
	Our idea is inspired by \cite[][Theorem 4]{FGMS20}.
	
	For each vertex $(i, x) \in [2^m] \times [2^n]$, we say that $(i, x)$ is {\it active} if it holds that $p'(x) \le i < p'(f'(x))$; the vertex $(i, x)$ is {\it inactive} if it is not active.
	For every inactive vertex $(i, x) \in [2^m] \times [2^n]$, we define the function $F(i, x) := (p(x), x)$.
	For each active vertex $(i, x) \in [2^m] \times [2^n]$, we define the function $F: [2^m] \times [2^n] \to [2^m] \times [2^n]$ as follows:
	\begin{align*}
		F(i, x) := \begin{cases}
			(i+1, x) & \text{ if } p'(x) \le i < p'(f'(x)) - 1, \\
			(p'(f'(x)), f'(x)) & \text{ if } i = p'(f'(x)) - 1.
		\end{cases}
	\end{align*}
	Furthermore, we define the potential function $P: [2^m] \times [2^n] \to [2^m]$ as follows: $P(i, x) = i$ if a vertex $(i, x)$ is active, otherwise $P(i, x) = p'(x) - 1$.
	It is not hard to see that every vertex satisfies the desired condition.
	We can effortlessly obtain a solution to the original instance from every solution to the new instance.
\end{proof}

We move on to describe how to construct the reduced \QuotientPigeonCircuit instance $\langle C: [2^n] \to [2^n], E: [2^n] \times [2^n] \to \{ 0, 1 \} \rangle$.
First, we define the equivalence relation $E: [2^n] \times [2^n] \to \{ 0, 1 \}$ with respect to the \LocalOPT instance $\langle f, p \rangle$.
For all $x, y \in [2^n]$, define $E(x, y) := 1$ if $p(x) = p(y)$; otherwise $E(x, y) = 0$.
It is straightforward to see that the function $E$ satisfies the requirements for the equivalence relation. There is no solution that holds the fourth, fifth, or sixth type of solution.
We define the Boolean circuit $C: [2^n] \to [2^n]$ as follows: $C(x) = f(x)$ for every $x \in [2^n]$.
Finally, we set the special element $v^*$ to be $1$.
We complete constructing a \QuotientPigeonCircuit instance $\langle C, E, v^* \rangle$.

What remains is to prove that we can recover a solution to the original \LocalOPT instance from each solution to the reduced \QuotientPigeonCircuit in polynomial time.
Since the Boolean circuit $E$ certainly computes the equivalence relation over $[2^n]$, every solution to the reduced instance $\langle C, E, v^* \rangle$ is one of the first-, second-, and third-type solutions.

\paragraph{First-type solution}
We first consider where we obtain two elements $x, y \in [2^n]$ such that $x \not\sim_{E} y$ and $C(x) \sim_{E} C(y)$.
This implies that $p(x) \neq p(y)$ but $p(f(x)) = p(f(x))$. Suppose that the point $x$ is not a solution to the \LocalOPT instance (i.e., $f(x) \neq x$ and $p(f(x)) = p(x) + 1$), the other point $y$ is a solution to the \LocalOPT instance. Similarly, suppose that $y$ is not a solution; the other point $x$ is a solution.
Hence, at least one of $x$ and $y$ is a solution to the \LocalOPT instance $\langle f, g \rangle$.

\paragraph{Second-type solution}
Next, we consider the case where we obtain an element $x \in [2^n]$ such that $C(x) \sim_{E} v^*$.
This implies that $1 = p(v^*) = p(f(x)) \le p(x)$. Therefore, the point $x$ is a solution to the \LocalOPT instance $\langle f, g \rangle$.

\paragraph{Third-type solution}
Finally, we consider the case where we obtain two distinct elements $x, y \in [2^n]$ such that $x \sim_{E} y$ and $C(x) \not\sim_{E} C(y)$.
This implies that $p(x) = p(y)$ but $p(f(x)) \neq p(f(y))$. From our assumption, exactly one of the points $x$ and $y$ is a fixed point of $f$. Therefore, we obtain a solution to the original \LocalOPT instance $\langle f, g \rangle$.

% Proof of PLC membership of Quotient Pigeon
\subsection{Proof of \cref{lemma:Q-Pigeon_to_LongChoice}} \label{proof:Q-Pigeon_to_LongChoice}
This section proves that the problem \QuotientPigeonCircuit belongs to the class $\PLC$.
To prove this, we will provide a polynomial-time reduction from \QuotientPigeonCircuit to \ConstrainedLongChoice, a restricted variant of the problem \LongChoice.
Our reduction heavily relies on the $\PPP$-hardness proof of \LongChoice by \citeauthor{PPY23} \cite[][Theorem 2]{PPY23}. \\

\begin{tcolorbox}
	\begin{definition}\label{def:constrainedlongchoice}
		\ConstrainedLongChoice
		
		\noindent
		\textbf{Input}: 
		$n - 1$ Boolean circuits $P_{0}, \dots, P_{n-2}$ such that $P_{i}: ([2^n])^{i+2} \to \{0, 1\}$ for each $i \in \{0, \dots, n-2\}$ 
		and an initial element $a_{0}$ in $[2^n]$
		
		\noindent
		\textbf{Output}: 
		a sequence of $n+1$ distinct elements $a_{0}, a_{1}, \dots, a_{n}$ in $[2^n]$ such that for each $i \in \{0, \dots, n-2\}$, $P_{i}(a_{0} \dots, a_{i}, a_{j})$ is the same for every $j > i$.
	\end{definition}
\end{tcolorbox}

\begin{proposition}[\citeauthor{PPY23} \cite{PPY23}]
	\LongChoice and \ConstrainedLongChoice are polynomial-time reducible to each other.
\end{proposition}

Let two Boolean circuits $C: [2^n] \to [2^n]$ and $E: [2^n] \times [2^n] \to \{0, 1\}$ and an element $v^* \in [2^n]$ be an instance of \QuotientPigeonCircuit.
From \cref{prop:no-fixedpoint}, we assume that $C(x) \not\sim_{E} x$ and $C(x) \neq v^*$ for every $x \in [2^n]$, without loss of generality. 
Also, from \cref{prop:assumption_for_quotient-pigeon}, we assume that the $2n$ elements $u_{0} := v^*, u_{1} := C(u_{0}), \dots, u_{i} := C(u_{i-1}), \dots, u_{2n} := C(u_{2n-1})$ are distinct under $\sim_{E}$ each other (i.e., $u_{i} \not\sim_{E} u_{j}$ for all $0 \le i < j \le 2n$), without loss of generality.
We will construct the instance of \ConstrainedLongChoice $\langle P_0, P_1, \dots, P_{n-2}, v^* \rangle$, where $P_{i}: ([2^n])^{i+2} \to \{ 0, 1 \}$ for every $i \in \{ 0, 1, \dots, n-2 \}$.

Before constructing the predicates $P_0, P_1, \dots, P_{n-2}$, we briefly sketch our reduction.
Every solution to \ConstrainedLongChoice is a sequence of distinct elements $a_0, a_1, \dots, a_n$ in $[2^n]$. 
In such a sequence, each element $a_{i}$ is chosen by depending only on the previously chosen elements $a_0, a_1, \dots, a_{i-1}$. 
Here, we use the terminology {\it distinct} to mean that $x \neq y$.
Recall that two distinct elements $x$ and $y$ such that $x \sim_{E} y$ and $C(x) = C(y)$ are not a solution to \QuotientPigeonCircuit.
So, we need to avoid such a {\it bad} solution being constructed as a solution to \ConstrainedLongChoice.
In order to settle such an issue, we arrange the sequence before applying the predicates $P_0, P_1, \dots, P_{n-2}$ that are defined in \cite{PPY23}.
We introduce the sub-procedures $\beta_{0}, \beta_{1}, \dots, \beta_{n-1}$, where for each index $k \in \{ 0, 1 \dots, n-1 \}$, the sub-procedure $\beta_{k}$ maps distinct $k+1$ elements $a_{0}, \dots, a_{k}, a_{k+1}$ to $k+1$ elements $b_{0}, \dots, b_{k}, b_{k+1}$ with a suitable property. 
%Note that these sequences may be the same.

Roughly speaking, we require that the elements $b_{0}, \dots, b_{n}$ are distinct under $\sim_{E}$ unless we can recover a solution to the original instance in polynomial time.
Each element $b_{k+1}$ depends only on $b_{0}, \dots, b_{k}$, and $a_{k+1}$.
Furthermore, if the sequence $(b_{0}, \dots, b_{n}, b_{n+1})$ contains two distinct elements $b_{i}, b_{j}$ such that  $b_{i} \sim_{E}, b_{j}$, then we can recover a solution to the original \QuotientPigeonCircuit instance $\langle C, E, v^* \rangle$ from $(b_{0}, \dots, b_{n}, b_{n+1})$ in polynomial time.

We first describe how to construct the sub-procedures $\beta_{0}, \dots, \beta_{n-1}$.
These sub-procedures are defined inductively.
We can find the formal structure for each index $k \in \{ 0, \dots, n-1 \}$ in \cref{alg:sub-procedure}.

Let $(a_0, \dots, a_n)$ denote an input sequence of the sub-procedures, where it holds that $a_{i} \neq a_{j}$ for all $0 \le i < j \le n$.
In the base case (i.e., $k = 0$), the sub-procedure $\beta_{0}$ directly outputs $b_{0}$ to be $a_{0}$.
Suppose that for an index $k < n - 1$, we have a sequence of elements $b_{0}, \dots, b_{k}$ constructed by the sub-procedure $\beta_{k - 1}$.
We now define the element $b_{k + 1}$ from the elements $b_{0}, \dots, b_{k}$, and $a_{k+1}$.

First, we check whether the elements $b_{0}, \dots, b_{k}$ are distcinct under $\sim_{E}$. If not, we set it to be $b_{k+1} := a_{k+1}$.
Otherwise, we also check whether $a_{k + 1} \not\sim_{E} b_{i}$ for every $i \in \{ 0, 1, \dots, k \}$. If yes, we define $b_{k+1} := a_{k+1}$.
Also otherwise, we have the elements $b_{0}, \dots, b_{k}$ and $a_{k+1}$ such that 
\begin{enumerate*}[label = (\arabic*)]
	\item $b_{i} \not\sim_{E} b_{j}$ for all $0 \le i < j \le n$; and
	\item there exists an elements $b_{i}$ such that $a_{k+1} \sim_{E} b_{i}$.
\end{enumerate*}
In this case, we verify whether we can recover a solution to the original \QuotientPigeonCircuit $\langle C, E, v^* \rangle$ from the elements $b_{0}, \dots, b_{k}$ and $a_{k+1}$ in polynomial time.
Specifically, we check the following six properties:
\begin{enumerate}
	\item There exist two elements $b_i$ and $b_j$ such that $b_{i} \not\sim_{E} b_{J}$ and $C(b_{i}) \sim_{E} C(b_{j})$;
	\item there is a elements $x$ in $\{ b_{0}, b_{1}, \dots, b_{k}, a_{k+1} \}$ such that $C(x) \sim_{E} v^*$;
	\item there is an element $b_{i}$ such that $b_{i} \sim_{E} a_{k+1}$ and $C(b_{i}) \not\sim_{E} C(a_{k+1})$;
	\item there exists an element $x$ in $\{ b_{0}, b_{1}, \dots, b_{k}, a_{k+1} \}$ such that $E(x, x) = 0$;
	\item there are two elements $x$ and $y$ in $\{ b_{0}, b_{1}, \dots, b_{k}, a_{k+1} \}$ such that $C(x, y) \neq E(y, x)$; and
	\item there exist distinct three elements $x, y, z \in \{ b_{0}, b_{1}, \dots, b_{k}, a_{k+1} \}$ such that $x \sim_{E} y$, $y \sim_{E} z$, and $x \not\sim_{E} z$.
\end{enumerate}
We call the algorithm that performs these above tests $\CheckSolutions$ (see also, \cref{alg:CheckSolutions}).

If it passes at least one of the above six tests, then we can efficiently recover a solution to the original \QuotientPigeonCircuit $\langle C, E, v^* \rangle$. Note that these tests can be computed in polynomial time.
Thus, we define $b_{k+1} := a_{k+1}$.
Finally, if the sequence of elements $(b_{0}, \dots, b_{k}, a_{k+1})$ is rejected by all of the above six tests, then we find the smallest positive integer $\ell$ such that $C^{\ell}(v^*) \not\sim_{E} x$ for every element $x \in \{ b_{0}, \dots, b_{k}, a_{k+1} \}$, and we define $b_{k+1} := C^{\ell}(v^*)$.
Note that such an integer $\ell$ always exists and is bounded by $2n$ from our assumption.

We complete constructing the sub-procedures $\beta_{0}, \dots, \beta_{n-1}$.
It is not hard to see that every sub-procedure $\beta_k$ is polynomial-time computable.
Moreover, a sequence of elements $b_{0}, \dots, b_{n}$ defined by our sub-procedures holds the next proposition.

\begin{proposition} \label{prop:sub-procedure}
	For every positive integer $k \in [n-1]$, and for every equence of elements $a_{0}, \dots, a_{k}, a_{k+1}$ in $[2^n]$ such that $a_{i} \neq a_{j}$ for all $0 \le i < j \le k+1$, if the sequence of elements $b_{0}, \dots, b_{k}, b_{k+1}$ that are produced by $\beta_{k}(a_{0}, \dots, a_{k}, a_{k+1})$ are not distinct under $\sim_{E}$ (i.e., there is a pair of elements $b_{i}$ and $b_{j}$ such that $b_{i} \sim_{E} b_{j}$), then the algorithm $\CheckSolutions(b_{0}, \dots, b_{k}, b_{k+1})$ returns \True.
\end{proposition}
\begin{proof}
	We prove this by induction.
	In the case where $k = 1$, the statement is trivial.
	
	Let $k$ be a positive integer in $[n-1]$, and let $a_{0}, \dots, a_{k}, a_{k+1}$ denote a sequence of elements in $[2^n]$ such that  $a_{i} \neq a_{j}$ for all $0 \le i < j \le k+1$.
	Also, we write $b_{0}, \dots, b_{k}, b_{k+1}$ for the corresponding sequence of elements.
	Suppose that the statement holds for every $i < k$; that is, the algorithm $\CheckSolutions(b_{0}, \dots, b_{i}, b_{i+1})$ returns \True if we have two elements $b_{j^1}$ and $b_{j^2}$ with $0 \le j^1 < j^2 \le i+1$ such that $b_{j^1} \sim_{E} b_{j^2}$.
	We will show that the statement also follows for the index $k$.
	
	We assume that there exist two elements $b_{i}$ and $b_{j}$ with $0 \le i < j \le k+1$ such that $b_{i} \sim_{E} b_{j}$.
	From the inductive supposition, we consider the case where we have an element $b_{j}$ such that $b_{j} \sim_{E} b_{k+1}$.
	By the constrcution of the sub-procedure $\beta_{k}$, the algorithm $\CheckSolutions(b_{0}, \dots, b_{k}, a_{k})$ returns \True because if not, the element $b_{k+1}$ is defined to be distinct with other elements under $\sim_{E}$.
	We set it to be $b_{k+1} := a_{k+1}$, and thus, the algorithm $\CheckSolutions(b_{0}, \dots, b_{k}, b_{k+1})$ returns \True.
\end{proof}

We move on to constructing the predicates $P_{0}, P_{1}, \dots, P_{n-2}$, where $P_{i}: ([2^n])^{i + 2} \to \{ 0, 1 \}$ for each $i = 0, 1, \dots, n-2$.
Our structure is straightforward: After applying our sub-procedures,  we apply the predictions defined by \citeauthor{PPY23} \cite{PPY23}.
For the self-containment, we now describe the definition of a sequence of finite sets $B_0, \dots, B_i, \dots$ and $F_0, \dots, F_i, \dots$.
We proceed with an inductive definition.

In the base case, we define $B_0 := [2^n] \setminus \{ v^* \}$.
Since $a_0 := v^*$ and $C(v^*) \not\sim_{E} v^*$, the unfilled set $\UnfilledSet{B_{0}}{a_0}{C}$ has size $2^n - 2$.
We define $F_{0} := B_{0}[\kappa]$, where $\kappa = 2^{n-1} - 1$ if $C(a_{0}) > 2^{n-1} - 1$; otherwise $\kappa = 2^{n-1}$.
Then, the finite set $F_0$ has size $2^{n-1} - 1$.

Suppose that we have a sequence of elements $a_0, \dots, a_k$ such that $a_0 = v^*$, and $B_i$ and $F_i$ are defined for every $i < k$.
Here, we denote by $b_0, \dots, b_k$ the sequence of elements that are outputs of $\beta_{k-1}(a_0, \dots, a_k)$.
We first define the finite set $B_k$ using the following rules:
\begin{enumerate}[label = (\Roman*)]
	\item If $B_{k-1}$ is a singleton, then $B_{k} = F_{k} = B_{k-1}$.
	\item Otherwise, if $C(b_k)$ belongs to $F_{k-1}$, then $B_{k} = F_{k-1}$. If $C(b_k)$ is not in $F_{k-1}$, then $B_{k} = B_{k-1} \setminus F_{k-1}$.
\end{enumerate}
Finally, we define $F_{k}$. 
Let $\kappa$ denote the smallest integer such that \[ | \UnfilledSet{B_{k}}{b_{0}, \dots, b_{k}}{C}[\kappa] | = \left \lceil \frac{|\UnfilledSet{B_{k}}{b_{0}, \dots, b_{k}}{C}|}{2} \right \rceil. \]
We define $F_{k} = B_{k}[\kappa]$.

From the above construction, we can see that $B_{0} \supseteq B_{1} \supseteq \cdots \supseteq B_{i} \supseteq \cdots \supseteq B_{n}$ and $F_{i} \subseteq B_{i}$ for every $i \in \{ 0, 1, \dots, n - 1 \}$.

To complete constructing the \ConstrainedLongChoice instance, we define the predicate function $P_{k}$ for each index $k$ in $\{ 0, 1, \dots, n-2 \}$ as follows:
\begin{align*}
	P_{k}(a_{0}, \dots, a_{k}, x) = \begin{cases}
		1 & \text{ if } C(b_{k+1}) \in F_{k}\\
		0 & \text{ if } C(b_{k+1}) \not\in F_{k},
	\end{cases}
\end{align*}
where $(b_0, \dots, b_{k+1}) = \beta_{k}(a_{0}, \dots, a_{k}, x)$.

We now obtain the reduced \ConstrainedLongChoice instance $\langle P_{0}, P_{1}, \dots, P_{n-2}, v^* \rangle$.
What remains is to prove that a feasible sequence for our instance allows us to recover a solution to the original \QuotientPigeonCircuit instance $\langle C, E, v^* \rangle$.

Let $a_0, a_1, \dots, a_n$ be a feasible sequence; that is, it holds that $a_{0} = v^*$, $a_{i} \neq a_{j}$ for all $0 \le i < j \le n$, and for each index $i \in \{0, \dots, n-2 \}$, $P_{i}(a_{0}, \dots, a_{i}, a_{j})$ are the same for every $j > i$.
Let $b_0, \dots, b_{n-1}, b_n$ denote the elements that are outputs of the sub-procedure $\beta_{n-1}(a_{0}, \dots, a_{n-1}, a_{n})$.
It suffices to show that $\CheckSolutions(b_0, \dots, b_{n-1}, b_{n})$ returns \True.

For the sake of contradiction, we suppose that $\CheckSolutions(b_0, \dots, b_{n-1}, b_{n})$ returns \False.
From \cref{prop:sub-procedure}, it satisfies that $b_{i} \not\sim_{E} b_{j}$ for all $0 \le i < j \le n$. That is, the elements $b_{0}, \dots, b_{n}, b_{n+1}$ are distinct.
Also, we have no element $x$ such that $C(x) \neq v^*$ from our assumption.
Furthermore, the following two properties hold.

\begin{proposition}[\citeauthor{PPY23} \cite{PPY23}] \label{prop:key_properties_PPY23}
	\begin{enumerate*}[label = (\arabic*)]
		\item \label{PPY23_property1} For each $i \in \{ 0, \dots, n \}$, $C(b_{j}) \in B_{i}$ for every $j > i$.
		\item \label{PPY23_property2} For every $i \in \{ 0, \dots, n-2 \}$, $| \UnfilledSet{B_{i}}{b_{0}, \dots, b_{i}}{C} | = 2^{n-i} - 2$.
	\end{enumerate*}
\end{proposition}

From \cref{PPY23_property2} of \cref{prop:key_properties_PPY23}, we have that $| \UnfilledSet{B_{n-2}}{b_{0}, \dots, b_{n-2}}{C} | = 2$.
Recall the definition of the subset $F_{n-2}$. It satisfies that $| \UnfilledSet{F_{n-2}}{b_{0}, \dots, b_{n-2}}{C} | = 1$ and $| \UnfilledSet{B_{n-2} \setminus F_{n-2}}{b_{0}, \dots, b_{n-2}}{C} | = 1$.
Since the $P_{n-2}(a_{0}, \dots, a_{n-2}, a_{n-1})$ and $P_{n-2}(a_{0}, \dots, a_{n-2}, a_{n})$ are the same, exactly one of the following holds: 
\begin{enumerate*}[label = (\roman*)]
	\item $C(b_{n-1})$ and $C(b_{n})$ are in $F_{n-2}$; and
	\item $C(b_{n-1})$ and $C(b_{n})$ are in $B_{n - 2} \setminus F_{n-2}$.
\end{enumerate*}
Therefore, we have a collision, which contradicts from $\CheckSolutions(b_0, \dots, b_{n-1}, b_{n})$ returns \False.

\section{Conclusion and Open Questions}
This short paper has investigated the computational aspects of \StonePicking.
We have shown the $\PLS$-hardness of \LongChoice, a $\TFNP$ formulation of \StonePicking.
Our result implies that the complexity class $\PLC$ also contains the class $\PLS$.
Furthermore, we have introduced the new $\TFNP$ problem \QuotientPigeonCircuit that is $\PPP$- and $\PLS$-hard as a by-product.

This short paper has left the following open questions:
\begin{enumerate}[label = (\roman*)]
	\item Does $\PLC$ contain $\PPA$? Thus, does $\PLC$ unify traditional $\TFNP$ subclasses?
	\item Is the problem \QuotientPigeonCircuit $\PLC$-hard? In other words, is \QuotientPigeonCircuit $\PLC$-complete?
	\begin{itemize}
		\item As a matter of course, a natural $\PLC$-complete problem is still unknown.
		\item We are also interested in the relationship between \QuotientPigeonCircuit and the problem \UnaryLongChoice, a $\TFNP$ formulation of the non-interactive variant of \StonePicking \cite[][Section 2]{PPY23}.
	\end{itemize}
\end{enumerate}

Finally, we remark that the concept of our $\TFNP$ problem \QuotientPigeonCircuit has come from \cite{Ish21} (and also \cite{HG18}).
\citeauthor{Ish21} \cite{Ish21} has shown the robustness of \EndofPotentialLine. To prove this, he has constructed a reduction by regarding several nodes of the original instance as one node of the reduced instance\footnote{Previously, \citeauthor{HG18} \cite{HG18} have proven the robustness of \EndofLine using the similar approach.}. Such an approach can be viewed as a quotient from the mathematical perspective. 
This short paper formulates a search problem over a quotient set by extending their ideas.
The primal purpose of this work has been to characterize the complexity of the variants with super-polynomially many known sources of \EndofLine and \EndofPotentialLine. It is still open whether such variants are also $\PPAD$- and $\EOPL$-complete, respectively.
We believe that a search problem on a quotient set helps us advance our understanding of $\TFNP$ problems.
For example, a computational problem related to the Chevalley-Warning theorem \cite{GKSZ20} is one of $\TFNP$ problems on quotient sets.

\section*{Acknowledgment}
This work was partially supported by JST, ACT-X, Grant Number JPMJAX2101.
I would like to appreciate Pavel Hub{\'{a}}cek for taking the time to talk about $\TFNP$ problems privately between September and October 2023.
During the discussion, I identified the future direction of my study; before that, I had been uncertain about the direction of my study.

\printbibliography

\clearpage
\appendix

\section{Procedures}

% Sub-procedures
\begin{algorithm}
\caption{The sub-procedure $\beta_{k}$ for $k \in \{ 0, 1, \dots, n \}$} \label{alg:sub-procedure}
\begin{algorithmic}[1]
\Require a sequence $(a_{0}, \dots, a_{k}, a_{k+1})$ on $[2^n]$ such that $a_{i} \neq a_{j}$ for all $0 \le i < j \le k+1$
\Ensure  a sequence $(b_{0}, \dots, b_{k}, b_{k+1})$ on $[2^n]$
\State $b_0 \leftarrow a_0$
\State Suppose that we have a sequence $(b_{0}, \dots, b_{k})$ by using the sub-procedures $\beta_0, \dots, \beta_{k-1}$ inductively.
\If{Exists a pair $b_i, b_j$ such that $b_i \sim_{E} b_j$}\label{algln2}
    \State $b_{k+1} \leftarrow a_{k+1}$
\ElsIf{$a_{k+1} \not\sim_{E} b_{i}$ for every $i \in \{ 0, 1, \dots, k\}$}
	\State $b_{k+1} \leftarrow a_{k+1}$
\ElsIf{$\CheckSolutions(b_0, \dots, b_k, a_{k+1})$ returns \True}
	\State $b_{k+1} \leftarrow a_{k+1}$
\Else
	\State Find the smallest positive integer $\ell$ such that $C^{\ell}(a_{0}) \not\sim_{E} b_{i}$ for every $i \in \{ 0, 1, \dots, k\}$
	\State $b_{k+1} \leftarrow C^{\ell}(a_{0})$
\EndIf
\end{algorithmic}
\end{algorithm}

% Check Solutions
\begin{algorithm}
\caption{The algorithm $\CheckSolutions$ that decides whether a solution to \QuotientPigeonCircuit exists} \label{alg:CheckSolutions}
\begin{algorithmic}[1]
\Require a sequence of elements $\xi_{0}, \dots, \xi_{k}, \xi_{k+1}$ in $[2^n]$
\Ensure Ether $\True$ or $\False$
\If{There exist two elements $\xi_{i}$ and $\xi_{j}$ such that $\xi_{i} \sim_{E} \xi_{j}$ and $C(\xi_{i}) \sim_{E} C(\xi_{j})$} % The first-type solution
\State \Return \True 
\ElsIf{There is an element $\xi \in \{ \xi_{0}, \dots, \xi_{k}, \xi_{k+1} \}$ such that $C(\xi) \sim_{E} v^*$} % The second-type solution
\State \Return \True
\ElsIf{There are two elements $\xi, \eta \in \{ \xi_{0}, \dots, \xi_{k}, \xi_{k+1} \}$ such that $\xi \sim_{E} \eta$ and $C(\xi) \not\sim_{E} C(\eta)$} % The third-type solution
\State \Return \True
\ElsIf{There exists an element $\xi \in \{ \xi_{0}, \dots, \xi_{k}, \xi_{k+1} \}$ such that $E(\xi, \xi) = 0$} % The fourth-type solution
\State \Return \True
\ElsIf{There are two elements $\xi, \eta \in \{ \xi_{0}, \dots, \xi_{k}, \xi_{k+1} \}$ such that $E(\xi, \eta) = E(\eta, \xi)$} % The fifth-type solution
\State \Return \True
\ElsIf{There exist distinct three elements $\xi, \eta, \zeta \in \{ \xi_{0}, \dots, \xi_{k}, \xi_{k+1} \}$ such that $\xi \sim_{E} \eta$, $\eta \sim_{E} \zeta$, and $\xi \not\sim_{E} \zeta$} % The sixth-type solution
\State \Return \True
\Else
\State \Return \False
\EndIf
\end{algorithmic}
\end{algorithm}

% Predicates
%\begin{algorithm}
%\caption{The predicate $P_{k}$ for $k \in \{ 0, 1, \dots, n - 2 \}$} \label{alg:our-predicates}
%\begin{algorithmic}[1]
%\Require a sequence $(a_{0}, \dots, a_{k}, a_{k+1})$ on $[2^n]$ such that $a_{i} \neq a_{j}$ for all $0 \le i < j \le k+1$
%\Ensure  Either $0$ or $1$
%
%\end{algorithmic}
%\end{algorithm}

\end{document}